%% file: main.tex
\documentclass[11pt]{article}
\pdfoutput=1
\usepackage{amsmath}
\usepackage{amssymb}
\usepackage{amsthm}
\usepackage{authblk}
\usepackage[mode=buildnew,subpreambles=true]{standalone}

\usepackage{algpseudocode}
\usepackage{algorithm}
\usepackage{tikz}
\usetikzlibrary{positioning}
\usepackage{fullpage} 
\usepackage{subcaption}

\graphicspath{{./includes/}}

\usepackage{hyperref}

\title{A Critique of Kumar's
``Necessary and Sufficient Condition for Satisfiability 
of a Boolean Formula in CNF and Its Implications on $\pe$ versus
$\np$ problem."}

\author{Michael C. Chavrimootoo\thanks{Supported in part by NSF grant 
		CCF-2006496.}\ } 
\author{Henry B. Welles\protect\footnotemark[1]}
\affil{Department of Computer Science\\University of Rochester\\Rochester, NY 14627, USA}

\date{December 11, 2021}

\newcommand{\condition}{\,\mid \:}
\newcommand{\naturalnumber}{\ensuremath{{\mathbb{N}}}}
\newcommand{\naturalnumberpositive}{\ensuremath{{\mathbb{N}^+}}}

\newcommand{\pe}{\mbox{\rm P}}
\newcommand{\np}{\mbox{\rm NP}}
\newcommand{\conp}{\mbox{\rm coNP}}
\newcommand{\sat}{\mbox{\rm SAT}}
\newcommand{\cnfsat}{\mbox{\rm CNF-SAT}}
\newcommand{\true}{\textit{true}}
\newcommand{\false}{\textit{false}}
\newcommand{\open}{\textit{open}}
\newcommand{\nll}{\textit{null}}
\renewcommand{\P}[1]{\ensuremath{\mathcal{P}(#1)}}
\newcommand{\F}{\ensuremath{{\cal F}}}
\newcommand{\V}{\ensuremath{{\cal V}}}

\newtheorem{theorem}{Theorem}

\begin{document}

\maketitle

\begin{abstract}
In this paper, we analyze the argument made by Kumar in the technical report
``Necessary and Sufficient Condition for Satisfiability of a Boolean Formula in 
CNF and Its Implications on $\pe$ versus $\np$ problem" \cite{kum:t:nec-suff}.
The paper claims to present a polynomial-time algorithm that decides $\cnfsat$\@. We 
show that the paper's analysis is flawed and that the fundamental underpinning of
its algorithm requires an exponential number of steps 
on infinitely many inputs.
\end{abstract}

\section{Introduction} \label{sec:intro}
We provide a summary and critique of the third version of Manoj Kumar's
technical report ``Necessary and Sufficient Condition for Satisfiability of 
a Boolean Formula in CNF and Its Implications on $\pe$ versus $\np$ problem"
\cite{kum:t:nec-suff}. The paper claims to have constructed an algorithm
that decides $\cnfsat$ in polynomial time.
In order to understand the significance of this claim, one must first 
understand the significance of $\np$-complete problems.

Securing cryptographic systems relies on the assumption that solving $\np$ 
problems is intractable.
For example, consider the problem of integer factorization, which has a 
language version in $\np \cap \conp$.
The security of RSA encryption relies on the
suspicion that integer factorization is 
intractable, even in the typical case.
On the other hand, problem modeling using $\sat$ (another $\np$ problem)
is very common in several 
areas of artificial intelligence. As a result, finding a polynomial-time 
algorithm for $\sat$ would benefit those areas 
tremendously. The question of whether all $\np$ problems can be solved in 
polynomial time is commonly referred to as the $\pe$ vs.\ $\np$ problem and is 
considered the most important unsolved problem in computational complexity 
theory, and arguably, in all of applied mathematics. $\np$-complete problems 
are significant because showing that one is in $\pe$ is enough to imply that 
$\pe=\np$ \cite{kar:b:reducibility}, thereby resolving the $\pe$ vs.\ $\np$ 
problem. Since $\cnfsat$ is an NP-complete problem \cite{kar:b:reducibility}, 
Kumar's purported algorithm to decide $\cnfsat$ in polynomial time would prove 
$\pe = \np$. However, the paper's analysis is deeply flawed and we will show 
that its algorithm runs in exponential time on an infinite number of 
inputs.

In Section~\ref{sec:prelim} we give the preliminaries necessary to
understand Kumar's paper. In Section~\ref{sec:summary}, we summarize its 
definitions and theorems and present the paper's central theorem and algorithm. 
Finally, in Section~\ref{sec:critique} we expose the flaw in the paper's 
algorithm, present an infinite family $\F$ of counterexamples that exploit 
the flaw, and review some optimizations to the algorithm (as proposed in 
Kumar's paper) and show that they have no significant impact on the algorithm's 
runtime if the input is in $\F$.

\section{Preliminaries} \label{sec:prelim}

Let $x$ be a Boolean variable. We call $x$ and $\overline{x}$ literals.
If the variable $x$ has value \true{} (\false), then the literal $x$ also 
has value \true{} (\false), while the literal $\overline{x}$ has value 
\false{} (\true). 
The two literals associated to a variable always have complementary
values and hence are called complementary literals.
A clause is a disjunction of literals and values (\true{} and \false).
For example, $C = (x_1 \lor x_2 \lor x_3)$ is a clause.
We say that a variable $x$ occurs in or appears in a clause if one of the 
two literals $x$ and $\overline{x}$ appears in the clause.
A Boolean formula in conjunctive normal form (CNF) is a conjunction of 
clauses. Such formulas are also called CNF formulas. For example, 
$F = (x_1 \lor x_2) \land (x_3 \lor \overline{x_2} \lor x_1) \land (x_1)$ is 
a CNF formula.
We say that a variable $x$ occurs in or appears in a CNF formula $F$ if $x$
occurs in at least one clause in $F$.
If $C$ is a clause, $V$ is a set of Boolean variables,
and all the variables of $C$ appear in $V,$ then $C$ is said to
be a clause over $V$\@. 
A CNF formula $F$ is said to
be a formula over $V$ if all of $F$'s clauses are over $V$. 
For example, if $V = \{x_1, x_2, x_3\}$, and $F = (x_1 \lor x_2) \land 
(x_3 \lor \overline{x_2} \lor x_1) \land (x_1)$, then $F$ is a Boolean 
formula over $V$.

Given a variable set $V$ and an assignment of values to the variables in $V$,
a clause (over $V$) is said to evaluate to \true{} if at least one of
its disjuncts has value \true{} under the assignment.
For example, given assignment $x_1 = \false,\ x_2 = \false$, the clause 
$C = (x_1 \lor \overline{x_2})$ evaluates to \true.  
On the other hand, given the assignment $x_1= \false,\ x_2 = \true$, 
$C$ does not evaluate to \true.
A CNF formula is satisfiable if there is an assignment of values to its 
variables such that, under the assignment, all of the formula's clauses evaluate to
\true. For example, the formula $F$ defined above is satisfiable as 
all of its clauses evaluate to \true{} under the
assignment $x_1=\true,\ x_2 = \false,\ x_3 = \true.$  
Notice that a formula can have 
multiple satisfying assignments.  The set of all satisfiable
CNF formulas is called $\cnfsat$, and the Satisfiability 
Problem, in this context, is the task of deciding whether
a given input is a satisfiable CNF formula.

Kumar's paper views CNF formulas and clauses as sets. 
A clause is viewed as a set of literals and values, and a CNF formula is a 
set of clauses.
For example, the formula $F$ defined above, would be written as
$\{ \{x_1, x_2\}, \{x_3, \overline{x_2}, x_1\}, \{x_1\}  \}$.
In this paper, we treat all CNF formulas and clauses as sets, just like 
Kumar's paper does.

Finally, we let $\naturalnumber = \{0, 1, 2, \ldots\}$ i.e., the set of
all natural numbers (including zero), and let 
$\naturalnumberpositive = \{1, 2, 3, \ldots\}$ i.e., the set of positive 
natural numbers.

\section{Understanding the Paper's Argument} \label{sec:summary}
Kumar's paper relates the problem of satisfiability to that of computing a 
set with a specific property, which we discuss in Section~\ref{sub:arg}.
We present the key theorems\footnote{
	We note that we do not state the paper's
	theorems as they originally appear. Rather, we provide equivalent 
	statements using more common notation.} 
that lead to this result in Section~\ref{sub:def}
and then provide an overview of the paper's
algorithm that purportedly decides $\cnfsat$ in polynomial time.

\subsection{Definitions and Concepts} \label{sub:def}
For the sake of simplicity, Kumar's paper ignores trivially satisfiable 
clauses, such as $C_1 \cup \{\true\}$, where $C_1$ is a clause, or
$C_2\,\cup \{x, \overline{x}\}$, where $C_2$ is a clause and $x$ is a 
variable, because those are easily recognizable in polynomial time.
The paper refers to such clauses as tautology clauses.
The set of all nontautology clauses over a fixed variable set
is called a complete formula. 
For example, if $V=\{x_1, x_2\}$, then the complete formula is
$F_2 = \{\{x_1, x_2\},$ $\{x_1, \overline{x_2}\},$ $\{\overline{x_1},
x_2\},$
$\{\overline{x_1},$ $\overline{x_2}\},$ $\{x_1\}, 
\{\overline{x_1}\}, \{x_2\},$ $\{\overline{x_2}\}, \emptyset\}$
\cite{kum:t:nec-suff}.
Note that $\emptyset$ is used to refer to the null clause,
which can never be satisfied and can be viewed as equivalent to
$\{\false\}$. 
Additionally, a clause $C$ over a variable set $V$ is
said to be fully populated over $V$ if every variable in $V$ 
occurs in $C$\@. 
If $A$ and $B$ are two unequal and fully populated clauses over 
$V$, then $A$ and $B$ are said to be sibling clauses.
Kumar's paper also defines the cardinality of a formula
$F$, denoted by $\lVert F\rVert$, to be the number clauses in
the formula. Finally, given a set $S$, $\P{S}$ denotes the powerset
of $S$. We list below, Theorems~\ref{t:compl} and~\ref{t:siblings},
which are used in the proof of Theorem~\ref{t:main}.
\begin{theorem}[\cite{kum:t:nec-suff}, Corollary 10.3]
\label{t:compl}
Let $V$ be a variable set and ${\cal S}$ be the set of all
fully populated clauses over $V$.
Then the complete Boolean formula $F$ (over $V$) can be written as
$F = \bigcup_{C \in {\cal S}} \P{C}.$
\end{theorem}
This theorem implies that given a variable set $V$, every clause over $V$ is 
a subset of some fully populated clause over $V$.

\begin{theorem}[\cite{kum:t:nec-suff}, Theorem 7.2]
\label{t:siblings}
Let $V$ be a variable set.
If $C_1$ and $C_2$ are two sibling clauses that are over $V$,
then $(\forall D_1\in \P{C_1}-\P{C_2})(\exists V' \subseteq V)
(\exists D_2\in\P{C_2})$
$[D_1$ and $D_2$ are sibling clauses that are over $V'].$
\end{theorem}
Informally, it means that given two sibling clauses $C_1$ and $C_2$,
one can construct a pair of sibling clauses, $D_1 \subseteq C_1$ 
and $D_2 \subseteq C_2$, such that $D_1$ and $C_2$ have no common literal.

\subsection{The Argument} \label{sub:arg}
By combining the above-mentioned definitions and theorems, 
the following theorem is proved:
\begin{theorem}[\cite{kum:t:nec-suff}, Theorem 10.8]
\label{t:main}
Let $V$ be a variable set.
A Boolean formula $F$ (over $V$) is satisfiable if and only if there exists 
fully populated 
clause $C$ (over $V$) such that $(\forall E\in \P{C})[E \not\in F]$.
\end{theorem}
\begin{proof}[Proof summary]
The ``if" part is proved by way of contradiction, using the assumption
that the consequent does not hold. The proof then demonstrates that,
given a formula $F$ that satisfies the antecedent,
(1) for each assignment $\alpha$ over $V$, there is a fully populated 
clause $C_\alpha$ that evaluates to \false{} under $\alpha$, 
(2) by the assumption, it follows that for each fully populated clause $\tilde{C}$,
there is a subset $E$ of $\tilde{C}$ such that $E \in F$, and
(3) given an assignment $\beta$ and a clause $\hat{C}$, if $\hat{C}$ evaluates to 
\false{} under $\beta$, then every subset of $\hat{C}$ also evaluates to \false{} 
under $\beta$.\footnote{Kumar's paper actually proves (1) and (3), 
but they're rather easy to see so we don't prove them here.}
For each assignment $\gamma$ over $V$, let $C_\gamma$ be a clause as described in
(1). It then follows from (2) and (3) that 
there is a subset of $C_\gamma$ that evaluates to \false{} and is a clause of $F$.
Therefore, $F$ is not satisfiable, which is a contradiction.

The ``only if" part is more complicated and makes use of Theorems
\ref{t:compl}~and~\ref{t:siblings}, given a clause $C$ that satisfies the 
antecedent, to show that for each clause 
$D \in F$, there is a clause $E \in \P{C}$ such that $D$ and $E$ are
sibling clauses. 
The next argument is that, given an assignment $\alpha$ and two sibling clauses,
at least one of the two clauses evaluates to \true{}
under assignment $\alpha$.
By picking an assignment $\beta$ such that $C$ evaluates to \false, it 
follows that every $E \in \P{C}$ also evaluate to \false{} under $\beta$
and that each $D$, which is a sibling clause of some $E \in \P{C}$, must then
evaluate to \true{} under $\beta$, thereby proving that $F$ is 
satisfiable.
\end{proof}

This theorem asserts that to show satisfiability of a Boolean formula $F$,
one only needs to find a fully populated clause $C$ such that $(\forall E\in 
\P{C})[E \not\in F]$. Failure to find such a clause implies
that $F$ is not satisfiable. 
To tackle this clause-finding problem, Kumar's paper constructs a tree based 
on the input CNF formula and searches for the clause in that tree.
We reproduce the paper's algorithm to build and search a tree, given a CNF 
formula, as Algorithm~\ref{alg:kumar}. 
To assist in this search, and in 
an attempt to prevent the tree from growing
exponentially large, Kumar's paper also implements a pruning algorithm.
We have also modified some of the paper's figures to produce 
Figure~\ref{fig:tree} to exemplify how the tree is constructed and pruned.

\begin{figure}[ht!]
	\centering
	\caption{The construction and pruning of the tree when 
		$F=\{\{\overline{x_1}\}, \{x_1, 
		\overline{x_2}\}\}$.}\label{fig:tree}
	\begin{subfigure}[b]{0.4\linewidth}
		\includestandalone{reprone}
		\caption{The subtree constructed from processing clause 
			$\{\overline{x_1}\}$.}
		\label{fig:reprone}
	\end{subfigure}
	\hspace*{0.1\linewidth}
	\begin{subfigure}[b]{0.4\linewidth}
		\includestandalone{pruneone}
		\caption{The result of pruning the previous subtree.}
		\label{fig:pruneone}
	\end{subfigure}
	
	\begin{subfigure}[b]{0.4\linewidth}
		\includestandalone{reprtwo}
		\caption{The resulting subtree after processing clause
			$\{x_1, \overline{x_2}\}$.}
		\label{fig:reprtwo}
	\end{subfigure}
	\hspace*{0.1\linewidth}
	\begin{subfigure}[b]{0.4\linewidth}
		\includestandalone{prunetwo}
		\caption{The resulting subtree after pruning the previous
			subtree.}
		\label{fig:prunetwo}
	\end{subfigure}	
\end{figure}

Construction begins with the empty tree that contains only the root node, which is a special node with 
only one child pointer.
All the other nodes in the tree have two child pointers. When those pointers 
are not assigned to a node, they can have one of two values: \textit{open} or
\textit{null}. The former indicates that a child can be added at that
location, while the latter indicates that no child can ever be added 
there.
By default, when a node is created, its
child pointers are set to \textit{open}.
Each node is labeled with a variable name. We shall say a variable $x$ is in 
the tree if there is a node in the tree with label $x$.
The left pointer out of a node labeled $x$ represents
the literal $\overline{x}$, while the right pointer represents literal
$x$. 
Kumar's paper treats these representations (of literals) as the labels of 
the pointers. Now, consider a path $p$ from the root downwards (optionally 
including an \textit{open} pointer, if there is one). 
The set of labels on the pointers in $p$ describes a unique clause. 
Let $V$ be the set of variables appearing in the tree.
Then, each path in the tree (from the root to a leaf) uniquely identifies a 
clause in the complete formula over $V$.
The main loop of the algorithm proceeds as follows.
For each $C$ in the input formula $F$, let the variables of $C$ be denoted
by $V_C$. For each $x \in V_C$, if $x$ is not in the tree, then each 
\textit{open} pointer is assigned to a new node labeled with $x$.
Once the iteration over $V_C$ ends, the tree is pruned.

Figure~\ref{fig:pruneone} (\ref{fig:prunetwo}) shows how the tree in Figure~\ref{fig:reprone} (\ref{fig:reprtwo}) is pruned by Algorithm~\ref{alg:kumar}. After a clause $C$ has been added to the tree, pruning consists of 
removing all paths representing supersets of $C$ from the tree. Given a 
clause $C$, to remove all of its 
supersets from the tree, it suffices to find the path from the
root that describes $C$. Let the last edge in that path be $e$ and the
corresponding pointer be $q$. Delete everything that is connected downwards of $q$, by setting $q$ to 
\textit{null}.
In the figure, the first
clause we look at is $\{\overline{x_1}\}$. The path that describes it is
\false, $\overline{x_1}$. Thus the left subtree of node $x_1$ 
(along with that edge) is deleted and the left pointer of that node is set
to $\nll$. The same process is used to prune the 
tree further upon encountering clause $\{x_1, \overline{x_2}\}$.
Since the complete clause $\{x_1, x_2\}$ ``survives" the pruning process, 
$F$ is satisfiable.

We note, before addressing the error in the algorithm, that while
the explanatory text and examples in Kumar's paper only prune the tree after 
it has been completely constructed, the actual code that is provided prunes 
the tree after each clause has been processed. 

\begin{algorithm}[h!]
\caption{Algorithm to purportedly decide CNF-SAT in polynomial time.}
\label{alg:kumar}
\begin{algorithmic}[1]
\Ensure Input $F$ is a CNF formula.
\State $V \leftarrow$ set of variables appearing in $F$.
\State $S \leftarrow$ root node with its pointer set to \open{}.
\ForAll{$C \in F$} \label{line:clauses}
    \If{$C = \emptyset$} \label{line:null}
        \State Reject.
    \EndIf
    \If{$C$ is not a \textit{tautology clause}} \label{line:taut}
        \ForAll{variables $v \in C$} \label{line:vars}
            \If{$S$ does not contain a node labeled $v$}
                \If{$S$ contains no \open{} pointers} \label{line:closed}
                    \State Reject.
                \EndIf
                \State Add a distinct node labeled
                $v$ to each \open{} pointer. \label{line:build}
            \EndIf
        \EndFor
        \ForAll{pointers $p$ in $S$}\label{line:prune}
            \If{the clause represented by the path from the root
            to $p$ supersets $C$}
                \State Set $p$ to \nll{} and delete all nodes below it.
            \EndIf
        \EndFor
    \EndIf
\EndFor
\If{$S$ contains no \open{} pointer}
    \State Reject.
\Else
    \State Accept.
\EndIf
\end{algorithmic}
\end{algorithm}

\section{Identifying the Error} \label{sec:critique}

The algorithm iterates over every clause $C \in F$, rejecting if $C = 
\emptyset$ and
ignoring $C$ if it is a tautology clause. If neither case is true, then the algorithm
iterates over all the variables in $C$ that are not in the tree $S$.
For each variable $x$, if there are no \open{} pointers
in $S$, then the algorithm rejects. Otherwise, $x$ is added to the tree. 
It's only on line~\ref{line:prune} that the algorithm starts to prune supersets
of $C$---this is the paper's crucial error.\footnote{
We note in passing that, technically, Kumar's algorithm accepts
the empty formula (i.e., $F=\emptyset$) when it should not. 
However,
correcting this error is rather trivial so we do not give that
error further consideration.}
Since the subtree construction takes place
in full (for a given clause) before pruning, 
there is the possibility that an exponentially large
tree will be produced before the algorithm has a chance to prune. 
To make matters worse, during the pruning step the algorithm iterates 
over every pointer in the tree, potentially iterating over an exponential 
number of pointers. We construct an infinite family of counterexamples below 
in which both issues occur.

\subsection{Counterexample}\label{sub:counter}
We now define the infinite family $\F$ of counterexamples to the proposed 
polynomial runtime of Algorithm~\ref{alg:kumar}. 
Without loss of generality, let the set of all variables be $\V = 
\{x_i \condition i \in \naturalnumberpositive \}.$ 
For each $n \in \naturalnumberpositive$, we define
$F_n = \{\{ x_j \in \V \condition j \in \naturalnumberpositive \land j \leq n\}\}$
and let $\F = \{ F_k \condition k \in \naturalnumberpositive \land k > 1 \}$ be 
our family of counterexamples.
Note that all Boolean formulas 
in $\F$ are satisfiable (although this has no bearing on the runtime of the 
algorithm).
We will now show how $\F$ precludes Algorithm~\ref{alg:kumar} from being
a polynomial-time algorithm.

Fix $F_n \in \F$. 
Since $\lVert F_n \rVert = 1$, the main loop of the algorithm only iterates 
once and solely inspects the single clause in $F_n$, $C$.
Because $C \neq \emptyset$ and $C$ is not a tautology clause, the algorithm 
will proceed to loop over all variables $x \in C$ on line \ref{line:vars}. 
By design, the algorithm won't begin pruning until line~\ref{line:prune}, at which 
point all variables in $C$ will have been added to the tree. 
$C$ will be the first clause the algorithm sees as $F_n$ contains only one 
clause. Since pointers are only set to $\nll$ during pruning, and the algorithm has yet to prune, all
pointers originating from leaf nodes will be \open{} and the check on
line~\ref{line:closed} will fail.
As $\lVert C \rVert = n$, the loop on line~\ref{line:vars} only runs $n$ 
times.
At the $i$th iteration of that loop, the algorithm will insert $2^{i-1}$ new 
nodes. Hence, after the loop has run $n$ times, the resulting tree will contain 
$1+\sum_{i=1}^n 2^{i-1} = 2^n$ nodes (including the root node). 
Intuitively, because no pruning takes
place, after the algorithm has looped through all $x \in C$ the tree will 
contain a branch representing every possible assignment to $C$. Hence, the 
size of this tree will be $2^n$.

Furthermore, since $C$ contains all the variables in $V$, the only paths that 
can represent supersets of $C$ are those that also contain all the variables in 
$V$. In fact, only one path in the tree can represent a superset of $C$: the 
path representing $C$ itself. This is because every other path must 
represent either a subset of $C$ or a sibling clause of $C$.
As a 
result, only the pointer at the end of the path representing $C$ would
be set to \nll{} during pruning and the size of the tree would still be 
exponential (specifically $2^n-1$). Even worse, to accomplish this the 
pruning step iterates over every pointer in the tree. Since the tree is 
exponentially large, line \ref{line:prune} will explore $2^{n+1}-1$ 
pointers.

We now argue that there is no polynomial that
upper bounds the runtime of Algorithm~\ref{alg:kumar}.
Without loss of generality, we can assume that there is a polynomial, say $g$, 
from $\naturalnumber$ to $\naturalnumber$
such that for each $n \in \naturalnumberpositive$,
the length of $F_n \in \F$
when given as an input to the algorithm is upper bounded by $g(n)$.
This is because one only needs roughly $\log_2(n)$ bits to represent each of the $n$
literals and a constant number of
extra bits to denote the separation between each pair of literals.
Let $q : \naturalnumber \rightarrow \naturalnumberpositive$
be the hypothesized polynomial that upper bounds the runtime of
Algorithm~\ref{alg:kumar}.
From our observations we have seen that, for each $n \in 
\naturalnumberpositive$ Algorithm~\ref{alg:kumar} performs at least
$2^{n+1}-1$ steps on input $F_n \in \F$. 
However, there exists a sufficiently large $n_0 \in \naturalnumberpositive$, such that 
for each $n \geq n_0$,
$2^{n+1}-1 > q(g(n))$.
Therefore, $q$ does not upper bound 
Algorithm~\ref{alg:kumar}'s runtime and so
there is no polynomial that upper bounds the runtime of 
Algorithm~\ref{alg:kumar}.

\subsection{Attempted Optimizations}\label{sub:opt}
Kumar's paper presents an algorithm to check for the existence of tautology 
clauses in polynomial time. However, all this algorithm does is check 
whether there are tautology clauses in the given Boolean formula. 
Because there is no $F \in \F$ that contains a tautology clause, 
these optimizations do not affect our family of counterexamples.

The paper also presents the following bounds (see 
\cite[Section 14]{kum:t:nec-suff}), which at first glance seem to 
help detect those formulas that can result in exponentially large trees.
However, as we will show, the bounds fail to cordon off all such Boolean 
formulas. Before reviewing the proposed bounds, new notation must be 
introduced. The number of clauses containing the literal $x$ in a Boolean 
formula $F$ will be denoted by 
$\#(F,x) = \lVert \{C \condition C \in F \wedge x \in C\} \rVert$.\footnote{
While this notation is slightly different from Kumar's, 
it is equivalent with less room for ambiguity 
(see \cite[Section 13.2]{kum:t:nec-suff}).}
Given a Boolean formula $F$ which contains no tautology clauses, and given that 
$F$ is over a variable set $V$ containing $n$ variables, the 
proposed bounds are:
\begin{enumerate}
    \itemsep0em
    \item \label{bound:cardinality}
    If $\lVert F \rVert > 3^n - 2^n$, then $F$ is unsatisfiable. 
    \item \label{bound:min}
    If $(\exists x \in V)[\text{min}(\#(F,x), \#(F,\overline{x})) > 3^{n-1} -
    2^{n-1}]$, then $F$ is unsatisfiable.
    \item \label{bound:false}
    If $(\exists x \in V)[\#(F,x) \leq 3^{n-1} - 2^{n-1} < \#(F,\overline{x})]$,
    then each satisfying assignment must map $x$ to \false.
    \item \label{bound:true}
    If ($\exists x \in V)[\#(F,\overline{x}) \leq 3^{n-1} - 2^{n-1} < \#(F,x)]$,
    then each satisfying assignment must map $x$ to \true.
\end{enumerate}
Fix $F_n \in \F$ and let 
$V$ be the set of variables appearing in $F_n$.
We will show that $F_n$ is not 
affected by these bounds. 
Since $\lVert F_n \rVert = 1$ and $n > 1$, bound \ref{bound:cardinality} 
does not apply.
Note that because $F_n$ only contains one clause, and because that clause
contains all variables in
$V$, for all $x \in V$ it is always the case that $\#(F_n,x) = 1$ and
$\#(F_n,\overline{x}) = 0$. This means that in bound~\ref{bound:min}, the 
min term
will always be min$(1, 0) = 0$ which is never greater than the bound of
$3^{n-1}-2^{n-1}$ when $n > 1$.
Thus $F_n$ is unaffected by bound~\ref{bound:min}. In
bound~\ref{bound:false}, the inequality will evaluate to 
$1 \leq 3^{n-1}-2^{n-1} < 0$ 
which is never true, hence $F_n$ is unaffected by bound~\ref{bound:false}. In
bound~\ref{bound:true} we have $0 \leq 3^{n-1}-2^{n-1} < 1$ which is never 
true as we require $n > 1$. Thus $F_n$ is unaffected by 
bound~\ref{bound:true}. Hence, all counterexamples in $\F$ are unaffected by 
these proposed bounds and Algorithm~\ref{alg:kumar} still runs in exponential 
time.

\section{Conclusion} \label{sec:conclusion}
Due to the oversight in Kumar's paper,
there is no polynomial that upper bounds the runtime of
Algorithm~\ref{alg:kumar},
even when the inputs are
Boolean formulas as simple as those in $\F$. Despite 
several of the paper's optimizations to the algorithm, 
the worst-case runtime remains unchanged. Thus Kumar's paper has not given a 
polynomial-time algorithm to decide $\cnfsat$, and so the paper fails to 
show $\pe = \np$ as claimed.

\section{Acknowledgments} \label{sec:ack}
We thank
Lane A. Hemaspaandra, 
Arian Nadjimzadah, and
David E. Narv\'{a}ez
for their comments on earlier drafts of this paper. All remaining errors are
the responsibility of the authors.

\bibliographystyle{alpha}
\bibliography{ref}

\end{document}

%% file: includes/reprone.tex
	\begin{tikzpicture}
		[-latex, auto, node distance =4 cm and 5cm, on grid,
		semithick,
		state/.style ={circle, draw, black, text=black, minimum width =1 cm},
		leaf/.style ={rectangle, draw, white, text=black, minimum width =1 
		cm}]
		\node [state] (root) at (0, 10) {\small{root}};
		\node [state] (A) at (0, 8) {$x_1$};
		\node [leaf] (B) at (-2, 6) {$\{\overline{x_1}\}$};
		\node [leaf] (C) at (2, 6) {$\{x_1\}$};
		\node[leaf] (empty) at (-0.5, 9) {\textit{false}};
		\draw (root) to (A); 
		\draw (A) -- (B) node [midway, above] 
		{\hspace{-2.2em}$\overline{x_1}$};
		\draw (A) -- (C) node [midway] {$x_1$};
	\end{tikzpicture}

%% file: includes/pruneone.tex
	\begin{tikzpicture}
		[-latex, auto, node distance =4 cm and 5cm, on grid,
		semithick,
		state/.style ={circle, draw, black, text=black, minimum width =1 cm},
		leaf/.style ={rectangle, draw, white, text=black, minimum width =1 
			cm}]
		\node [state] (root) at (0, 10) {\small{root}};
		\node [state] (A) at (0, 8) {$x_1$};
		\node [leaf] (C) at (2, 6) {$\{x_1\}$};
		\node[leaf] (empty) at (-0.5, 9) {\textit{false}};
		\draw (root) to (A); 
		\draw (A) -- (C) node [midway] {$x_1$};
	\end{tikzpicture}

%% file: includes/reprtwo.tex
	\begin{tikzpicture}
		[-latex, auto, node distance =4 cm and 5cm, on grid,
		semithick,
		state/.style ={circle, draw, black, text=black, minimum width =1 cm},
		leaf/.style ={rectangle, draw, white, text=black, minimum width =1 
		cm}]
		\node [state] (root) at (0, 10) {\small{root}};
		\node [state] (A) at (0, 8) {$x_1$};
		\node [state] (C) at (2, 6) {$x_2$};
		\node [leaf] (CL) at (1, 4) {$\{x_1, \overline{x_2}\}$};
		\node [leaf] (CR) at (3, 4) {$\{x_1, x_2\}$};
		\node[leaf] (empty) at (-0.5, 9) {\textit{false}};
		\draw (root) to (A); 
		\draw (A) -- (C) node [midway] {$x_1$};
		\draw (C) -- (CL) node [midway, above] 
		{\hspace{-2.2em}$\overline{x_2}$};
		\draw (C) -- (CR) node [midway] {$x_2$};
	\end{tikzpicture}

%% file: includes/prunetwo.tex
	\begin{tikzpicture}
		[-latex, auto, node distance =4 cm and 5cm, on grid,
		semithick,
		state/.style ={circle, draw, black, text=black, minimum width =1 cm},
		leaf/.style ={rectangle, draw, white, text=black, minimum width =1 
			cm}]
		\node [state] (root) at (0, 10) {\small{root}};
		\node [state] (A) at (0, 8) {$x_1$};
		\node [state] (C) at (2, 6) {$x_2$};
		\node [leaf] (CR) at (3, 4) {$\{x_1, x_2\}$};
		\node[leaf] (empty) at (-0.5, 9) {\textit{false}};
		\draw (root) to (A); 
		\draw (A) -- (C) node [midway] {$x_1$};
		\draw (C) -- (CR) node [midway] {$x_2$};
	\end{tikzpicture}